\documentclass[floats,12pt,onecolumn]{iopart}
\usepackage{amsfonts}
\usepackage{calc}
\usepackage{subfigure}
\usepackage{graphicx} 
\usepackage{enumerate}
\usepackage{amssymb}
\usepackage{amsthm}
\newcommand{\appref}[1]{\hyperref[#1]{{Appendix~\ref*{#1}}}}
\newcommand{\cancel}[1]{} 
\newcommand*{\cA}{\mathcal{A}} 
\newcommand*{\cB}{\mathcal{B}}
\newcommand*{\cC}{\mathcal{C}}

\newcommand*{\cL}{\mathcal{L}}

\newcommand*{\cD}{\mathcal{D}}
\newcommand*{\cQ}{\mathcal{Q}}

\newcommand*{\cP}{\mathcal{P}}
\newcommand*{\cS}{\mathcal{S}}

\newcommand*{\cT}{\mathcal{T}}

\renewcommand*{\tr}{\mathop{\mathrm{tr}}\nolimits}
\newcommand*{\sgn}{\mathop{\mathrm{sgn}}\nolimits}

\newcommand{\bc}{\begin{center}}
\newcommand{\ec}{\end{center}}

\newcommand{\id}{\mathbb{I}}

\newtheorem{theorem}{Theorem}
\newtheorem{lemma}{Lemma}

\usepackage{amsfonts}

\def\id{\mathbb{I}}

\def\01{\{0,1\}}

\newcommand{\ket}[1]{|#1\rangle}
\newcommand{\bra}[1]{\langle#1|}

\newcommand{\proj}[1]{|#1\rangle\langle#1|}

\newcommand*{\eqref}[1]{(\ref{#1})}

\begin{document}
\newcommand*{\spec}{\mathsf{spec}}
\newcommand*{\sspec}{\mathsf{Sspec}}
\newcommand*{\photonnumber}{{\bf N}}
\newcommand*{\Mat}{\mathsf{Mat}}
\newcommand*{\poi}{\mathsf{Poi}}
\newcommand*{\bin}{\mathsf{Bin}}

\newcommand*{\bT}{\mathcal{T}_{\textsf{av}}}
\renewcommand*{\id}{\mathsf{id}}
\renewcommand*{\P}{{\bf P}}
\newcommand*{\hit}{{\bf t}}
\newcommand*{\gs}{\mathsf{Q}} 
\newcommand*{\gsc}{\mathsf{Q}^\bot} 

\newcommand{\vanyondn}{\raisebox{-0.5em}{\includegraphics[width=1.5em]{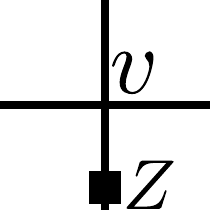}}}
\newcommand{\vanyonle}{\raisebox{-0.5em}{\includegraphics[width=1.5em]{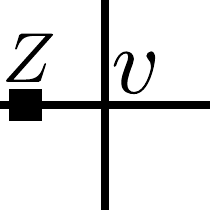}}}
\newcommand{\panyonup}{\raisebox{-0.5em}{\includegraphics[width=1.5em]{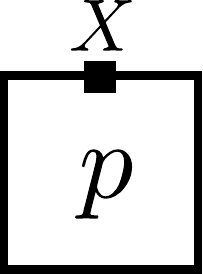}}}
\newcommand{\panyonri}{\raisebox{-0.5em}{\includegraphics[width=1.5em]{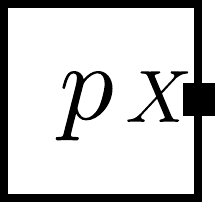}}}

\newcommand{\vertexstabilizer}{\raisebox{-0.75em}{\includegraphics[width=2.5em]{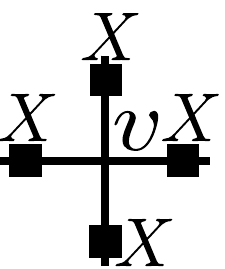}}}
\newcommand{\plaquettestabilizer}{\raisebox{-0.75em}{\includegraphics[width=2.5em]{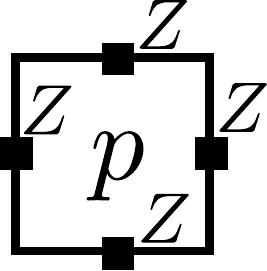}}}

\title{An optimal  dissipative encoder for the toric code}
\author{John Dengis$^1$, Robert K\"onig$^1$ and Fernando Pastawski$^2$}

\address{$^1$ Institute for Quantum Computing and Department of Applied Mathematics, University of Waterloo, 200 University Ave W, Waterloo, ON N2L 3G1, Canada\\
$^2$ Institute for Quantum Information and Matter, California Institute of Technology, 1200 E California Blvd, Pasadena, CA 91125, USA
}


\begin{abstract}
We consider the problem of preparing specific encoded resource states for the toric code 
 by local, time-independent interactions with a memoryless environment.  We propose a construction of such a dissipative encoder which converts product states to topologically ordered ones while preserving logical information. The corresponding Liouvillian is made up  of four-local Lindblad operators. For a qubit lattice of size~$L\times L$, we show that this process prepares encoded states in time~$O(L)$, which is optimal. This scaling compares favorably with known local unitary encoders for the toric code which take time of order~$\Omega(L^2)$ and require active time-dependent control. 
\end{abstract}\pacs{03.67.-a, 03.65.Vf}
\maketitle

\section{Introduction and main result}
Dissipation, while generally seen as detrimental for quantum computers, can nevertheless be a useful resource if suitably engineered. Appropriately chosen system-bath couplings can result in a non-equilibrium dynamics where initial states converge towards some dynamical steady state. This kind of  `quantum reservoir engineering' has been proposed as a viable approach towards the experimental preparation of interesting many-body states~\cite{KrausDiehlZoller08,verstraetewolfcirac09}. Remarkable examples include the preparation of pure states with long-range order in Bose-Einstein-condensates~\cite{DiehlMicheli08}, photonic arrays~\cite{Tomadinetal12}, as well as 
topologically ordered states~\cite{verstraetewolfcirac09,Bardynetal13}.  More generally, Verstraete et al.~\cite{verstraetewolfcirac09} argued that, at least in principle, an arbitrary quantum computation can be realized by dissipation. The corresponding process 
is similar to Feynman's clock construction and has the final state of the computation as its steady state. Subsequent work~\cite{Katoryanotiming} following this program proposed dissipative gadgets allowing to realize different dissipative dynamics during subsequent time-intervals.

Here we examine the dissipative preparation of specific topologically ordered states. 
While not realizing a fully dissipative computation, this basic primitive could act as a  building block in a hybrid scheme where initial states for quantum computation are prepared by thermalization and subsequent computations are performed in the usual  framework of topological quantum computation. We ask whether dissipative processes can be used to realize an {\em encoder}, i.e., a map which turns states on individual physical qubits  into encoded  (many-qubit) states. In contrast, previous work only considered the dissipative preparation of {\em some} ground state without guarantees  on the logical information.

It is worth mentioning that various unitary encoders are known for topologically ordered systems. 
For the toric code~\cite{kitaev97} on an $L\times L$~lattice, Dennis et al.~\cite{dennisetal02} gave a unitary circuit with two-local controlled-not (CNOT) gates of depth~$\Theta(L^2)$ acting as an encoder.   
Bravyi et al.~\cite{Bravyietalpropagationtop}  showed for  any evolution under a local time-dependent Hamiltonian acting as an encoder requires time at least~$\Omega(L)$.
In turn, Brown et al.~\cite{Brown2011} present a duality transformation from a 2D cluster state to a topologically ordered state which can be interpreted as a geometrically local quantum circuit of matching depth.
Dropping the requirement of locality, Aguado and Vidal~\cite{AguadoVidal08} gave an encoder with depth~$O(\log L)$ with geometrically non-local two-qubit gates.

The dissipative encoder considered here may be realized by designing suitable system-environment interactions. This is to be contrasted with schemes involving error correction, which generally consist of syndrome extraction by measurement and associated correction operations. For the toric code, an encoding procedure of this form was given~\cite{Grudkaetal12}. It involves active error correction operations similar to the minimal matching technique used in~\cite{dennisetal02}.

To define the notion of an encoder in more detail, consider a quantum error-correcting code $\cQ\cong(\mathbb{C}^2)^{\otimes k}\subset(\mathbb{C}^2)^{\otimes n}$ encoding $k$~logical qubits into $n$~physical qubits. Informally, an encoder is a map taking any state~$\ket{\Psi}\in(\mathbb{C}^2)^{\otimes k}$ into its encoded version~$\ket{\overline{\Psi}}\in\cQ\subset (\mathbb{C}^2)^{\otimes n}$. (This notion implicitly assumes a choice of basis of~$\cQ$). Since we are interested in a physical system of $n$~qubits, we will require the encoder to convert  a `simple' unencoded initial state into an encoded logical state. That is, we 
ask that for a fixed subset $A_1,\ldots,A_k$ of qubits and a fixed product state $\ket{\varphi}_{A_{k+1}}\otimes\cdots\otimes \ket{\varphi}_{A_n}$ on the remaining qubits, the encoder maps
\begin{eqnarray}
\ket{\Psi}_{A_1\cdots A_k}\otimes \ket{\varphi}_{A_{k+1}}\otimes\cdots\otimes \ket{\varphi}_{A_n}\quad\mapsto \quad \ket{\overline{\Psi}}\in\cQ \qquad\textrm{ for all }\ket{\Psi}\in(\mathbb{C}^2)^{\otimes k}\ .\label{eq:encodermap}
\end{eqnarray}
We are interested in encoders realized by evolution under a Markovian master equation
\begin{eqnarray*}
\frac{d}{dt}\rho=\cL(\rho)\ .
\end{eqnarray*}  Here the Liouvillian has Lindblad form
\begin{eqnarray*}
\cL(\rho) &=\sum_j L_j\rho L_j^\dagger -\frac{1}{2}\{L_j^\dagger L_j,\rho\}\ 
\end{eqnarray*}
with Lindblad operators~$L_j$ acting locally  on a constant number of qubits. 
We ask whether the completely positive trace-preserving map (CPTPM)~$e^{t\cL}$ generated by~$\cL$ is an (approximate) encoder for sufficiently large times~$t$, i.e., whether it can realize the map~\eqref{eq:encodermap}.  Our result is the following.
\begin{theorem}\label{thm:encoder}
Let $\cQ\subset (\mathbb{C}^2)^{\otimes 2L^2}$ be the toric code consisting of $2L^2$ qubits on the edges of a periodic $L\times L$-lattice.
 Consider the partition of the qubits into disjoint sets $\cA\cup\cB\cup \cB'\cup \cC\cup\cC'\cup\cD$
shown in Figure~\ref{fig:encoderqubits}. That is,  
\begin{itemize}
\item
$\cA=\{A_1,A_2\}$  are two neighboring qubits having a common adjacent vertex~$v_*$ and plaquette $p_*$.
\item
$\cB=\{B_1,\ldots,B_{L-1}\}$ and $\cC=\{C_1,\ldots,C_{L-1}\}$ are located along  a vertical respectively horizontal line passing through $A_1$.
\item
$\cB'=\{B'_1,\ldots,B'_{L-1}\}$ and $\cC'=\{C'_1,\ldots,C'_{L-1}\}$ are located along  a horizontal respectively vertical line passing through $A_2$.
\item
$\cD$ are the remaining $2(L-1)^2$ qubits.
\end{itemize}
There is a geometrically local Liouvillian~$\cL$ (with $4$-qubit Lindblad operators)
such that the following holds:
For any state $\rho_\cD$  on $\cD\cong(\mathbb{C}^{2})^{\otimes 2(L-1)^2}$ and $\ket{\Psi}\in(\mathbb{C}^2)^{\otimes 2}$, and for any  $\epsilon>0$, we have
\begin{eqnarray*}
\Big\|e^{t\cL}\left(\proj{\Psi}_{\cA}\otimes\proj{+}^{\otimes 2(L-1)}_{\cB\cB'}\otimes\proj{0}^{\otimes 2(L-1)}_{\cC\cC'}\otimes\rho_{\cD}\right)-\proj{\overline{\Psi}}\Big\|_1\leq\epsilon\ .
\end{eqnarray*} 
whenever
\begin{eqnarray}
t\geq  (4\ln(2))\cdot L+2\ln(16\epsilon^{-2})\ .
\end{eqnarray}
In this expression, we use the trace norm $\|A\|_1=\tr\sqrt{A^\dagger A}$  for Hermitian operators. 
\end{theorem}
\noindent We will give a detailed description of the relevant Liouvillian in Section~\ref{sec:schemedescription}. 
\begin{figure}
\centering
\includegraphics[width=10cm]{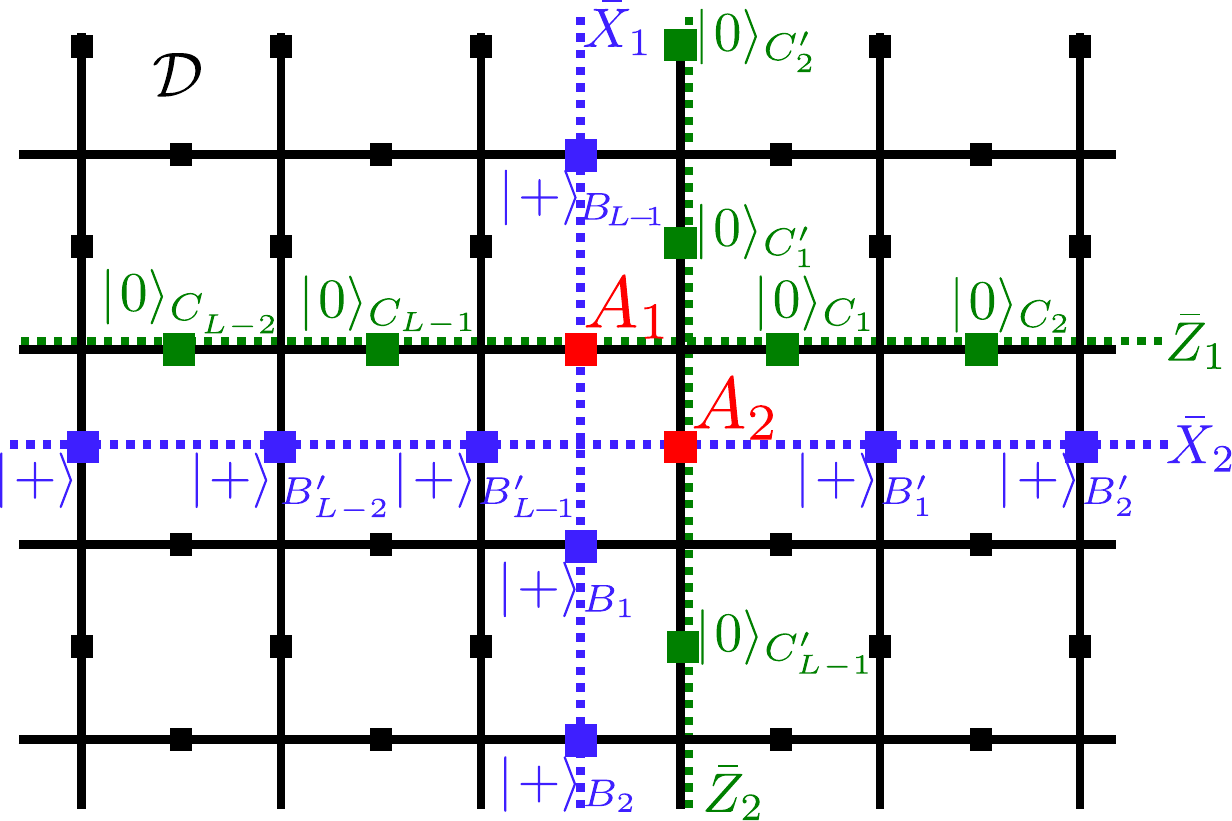}
\caption{This figure indicates the relevant qubits in Theorem~\ref{thm:encoder}, and their appropriate initialization for encoding: Qubits $A_1,A_2$ are initialized in the state to be encoded. 
Each qubit in~$\cB\cup\cB'$ is in the state~$\ket{0}$, while  qubits in~$\cC\cup\cC'$ are in the state~$\ket{+}$. 
The state of the remaining qubits~$\cD$ can be an arbitrary (mixed) state. 
In the figure, we also illustrate the support of possible realizations for logical operators $(\bar{X}_1,\bar{Z}_1)$ and $(\bar{X}_2,\bar{Z}_2)$ associated with the first and second logical qubit, respectively.
Our encoding procedure requires choosing one such realization for each logical Pauli generator such that they overlap only at the initial unencoded qubits.
\label{fig:encoderqubits} }
\end{figure}
The evolution~$e^{t\cL}$ implements a continuous-time version of a local error correction process somewhat analogous to Toom's rule: excitations move towards a single plaquette/vertex where they annihilate. 
An analogous ground-state preparation scheme for more general frustration-free Hamiltonians was discussed in~\cite{verstraetewolfcirac09}. However, guarantees about encoded information appear to be harder to obtain in their generic scheme. 
Furthermore because their construction requires the injectivity~\cite{PEPSinjective} property of the associated projected entangled pair state (PEPS) description (and hence blocking of sites), the resulting locality of the Lindblad operators will be slightly worse. In contrast, our scheme directly exploits the stabilizer structure of the underlying code, resulting in a comparatively simple Liouvillian.  
Indeed, our construction is optimal in  terms of locality, i.e., the number of particles involved in each Lindblad term~\cite{Marvian2013}. Related locality constraints for pure steady states were  derived\footnote{In terms of \cite[Definition 2]{Ticozzi2014}, our construction is capable of preparing any state in  ground space of the toric code as a \emph{conditionally asymptotically stable state} of a Lindblad dynamics.} in~\cite{TicozziViola12, Ticozzi2014}.  In our setup,  the initial product state determines in a transparent way which code state is prepared.  Our work goes beyond the mere characterization of steady states by adding two key ingredients: the consideration of logical observables (which are preserved) and an analysis of the convergence towards the ground space of the toric code.

The bound $O(L)$ on the convergence time
established by Theorem~\ref{thm:encoder} improves on  the $O(L^2)$ upper bound predicted for the analogous construction  in~\cite{verstraetewolfcirac09}, without a guarantee on the logical information. In fact, it is tight:  there are initial states~$\rho$ 
which thermalize slowly, i.e.,~$e^{t\cL}(\rho)$ is  far away from the code space~$\cQ$ for any time~$t\ll L$. In separate work~\cite{KoePas13}, we provide a general no-go theorem in this direction: dissipative state preparation of topologically ordered states requires at least a linear amount of time in~$L$ if the Liouvillian is local.  Combined with Theorem~\ref{thm:encoder}, this implies that the construction presented here is optimal in terms of preparation time among the entire class of local Liouvillians.  

In summary, our work shows that dissipative processes can be used to implement an encoder for the  toric code. 
Intriguingly, this encoder is more time-efficient than the best known unitary circuit. 
We stress, however, that  both types of encoders need to be supplemented with additional mechanisms in the presence of noise, especially if the encoded information is further processed. As discussed in~\cite{Pastawskietalmemory11}, local Liouvillians such as the one considered here are not suitable for the preservation of encoded information.

\section{Description of the Liouvillian\label{sec:schemedescription}}
\subsection{A Liouvillian for a general stabilizer code\label{sec:generalstabilizercodes}}
We first describe a generic Liouvillian associated with a stabilizer code~$\cQ$ with stabilizer generators~$\{S_j\}_{j\in\cS}$. 
We will subsequently specialize this to the toric code. Let~$\P_j^{\pm}=\frac{1}{2} (I\pm S_j)$ be the projections onto the $\pm 1$ eigenspaces of the stabilizer~$S_j$. 
The code space~$\cQ$ is the ground space of the Hamiltonian
\begin{eqnarray*}
H&=\frac{1}{2}\left(|\cS|\cdot I-\sum_j S_j\right)=\sum_{j}\P_j^-
\end{eqnarray*}
(the global energy shift is introduced for convenience). 
Our goal is to implement a local error-correction strategy by a dissipative evolution.
Concretely, we associate a unitary Pauli correction operator~$C_j$ to each stabilizer generator~$S_j$ such that $C_j$ and $S_j$ anticommute, i.e.,  $\{C_j, S_j\}=0$.  
 We define the CPTPM
\begin{eqnarray*}
\cT_j (\rho) &=\P_j^+\rho\P_j^+ + C_j\P_j^-\rho\P_j^- C_j^\dagger\ .
\end{eqnarray*}
Note that by definition,~$\cT_j$ lowers the energy of the term $\P^-_j$ in the Hamiltonian, that is, 
\begin{eqnarray}
\tr(\cT_j(\rho)\P^-_j)\leq \tr(\rho \P^-_j)\qquad\textrm{ for any state }\rho\ .\label{eq:loweringenergy}
\end{eqnarray}
While after application of~$\cT_j$, the stabilizer constraint defined by~$S_j$ is satisfied, its application may create a non-trivial syndrome (excitation) for a neighboring stabilizer~$S_{k}$, $k\neq j$. 
By design (i.e., the choice of correction operators for the toric code discussed below),  repeated application of all~$\{\cT_j\}_{j\in\cS}$ eventually removes all excitations, resulting in a state supported on the code space~$\cQ$. It will be convenient to introduce the averaged CPTPM
\begin{eqnarray}
\bT(\rho)&=\frac{1}{|\cS|} \sum_j \cT_j\ ,\label{eq:averagebtsuperop}
\end{eqnarray}
which randomly chooses a syndrome and applies the associated correction map.

By the correspondence discussed in~\cite{WolfCirac06}, 
each local CPTPM $\cT_j$ defines a local Liouvillian $\cL_j=\cT_j-\id$. 
We are interested in the evolution under $\cL=\sum_j \cL_j$, which, by definition, is a sum of constant-strength local terms. 
 Observe that $\cL=|\cS|\cdot (\bT-\id)$.

\subsection{Construction for the toric code \label{sec:toriccodedescr}}
We now consider the toric code with qubits on the edges of an $L\times L$ (periodic) square lattice.  We  separate the  Hamiltonian into
\begin{eqnarray*}
\fl H&=H^{(p)}+H^{(v)}\ \textrm{ where }H^{(p)}=\frac{1}{2} \left(L^2\cdot I-\sum_{p\in\cS^{(p)}} S_p\right)\ \textrm{ and }\ H^{(v)}=\frac{1}{2}\left(L^2\cdot I-\sum_{v\in\cS^{(v)}} S_v\right)\ .
\end{eqnarray*}
where the former includes all plaquette- and the latter includes all vertex terms.
Here, we have taken $S_p$ and $S_v$ to correspond to plaquette and vertex stabilizers respectively. 
\begin{equation}
S_p=Z^{\otimes 4}=\plaquettestabilizer \qquad\qquad S_v=X^{\otimes 4}=\vertexstabilizer,
\end{equation}
In other words, $S_p$ is the product of $Z$-type Pauli operators acting on the four qubits on the edges bounding plaquette~$p\in\cS^{(p)}$ and $S_v$ is the product of four $X$-type Pauli operators acting on the qubits incident to vertex~$v\in\cS^{(v)}$.
To define the associated Pauli correction operators $\{C_v\}_{v\in\cS^{(v)}}$ and $\{C_p\}_{p\in\cS^{(p)}}$, let us partition the set of vertices and plaquettes into
\begin{equation}
\cS^{(p)}=\{p_*\}\cup \cS^{(p)}_\rightarrow\cup \cS^{(p)}_\uparrow\qquad\textrm{ and }\qquad\cS^{(v)}=\{v_*\}\cup \cS^{(v)}_\leftarrow\cup \cS^{(v)}_\downarrow\  . 
\end{equation}
Here $p_*$ is a single plaquette,~$\cS^{(p)}_{\rightarrow}$ consists of the $L-1$~plaquettes lying on a fundamental cycle of the torus (which we refer to the `equator', running `horizontally' or `east-west' along the torus) on which~$p_*$ is located, whereas $\cS^{(p)}_{\uparrow}$ are the remaining~$L^2-L$~plaquettes. The vertex~$v_*$ as well as the sets $\cS^{(v)}_\leftarrow$ and $\cS^{(v)}_\downarrow$ are defined similarly on the dual lattice, 
see Figure~\ref{fig:liouvilleandef}.

\begin{figure}
\centering
\subfigure[Correction operations for plaquettes\label{it:plaquettesfig}]{\includegraphics[width=8cm]{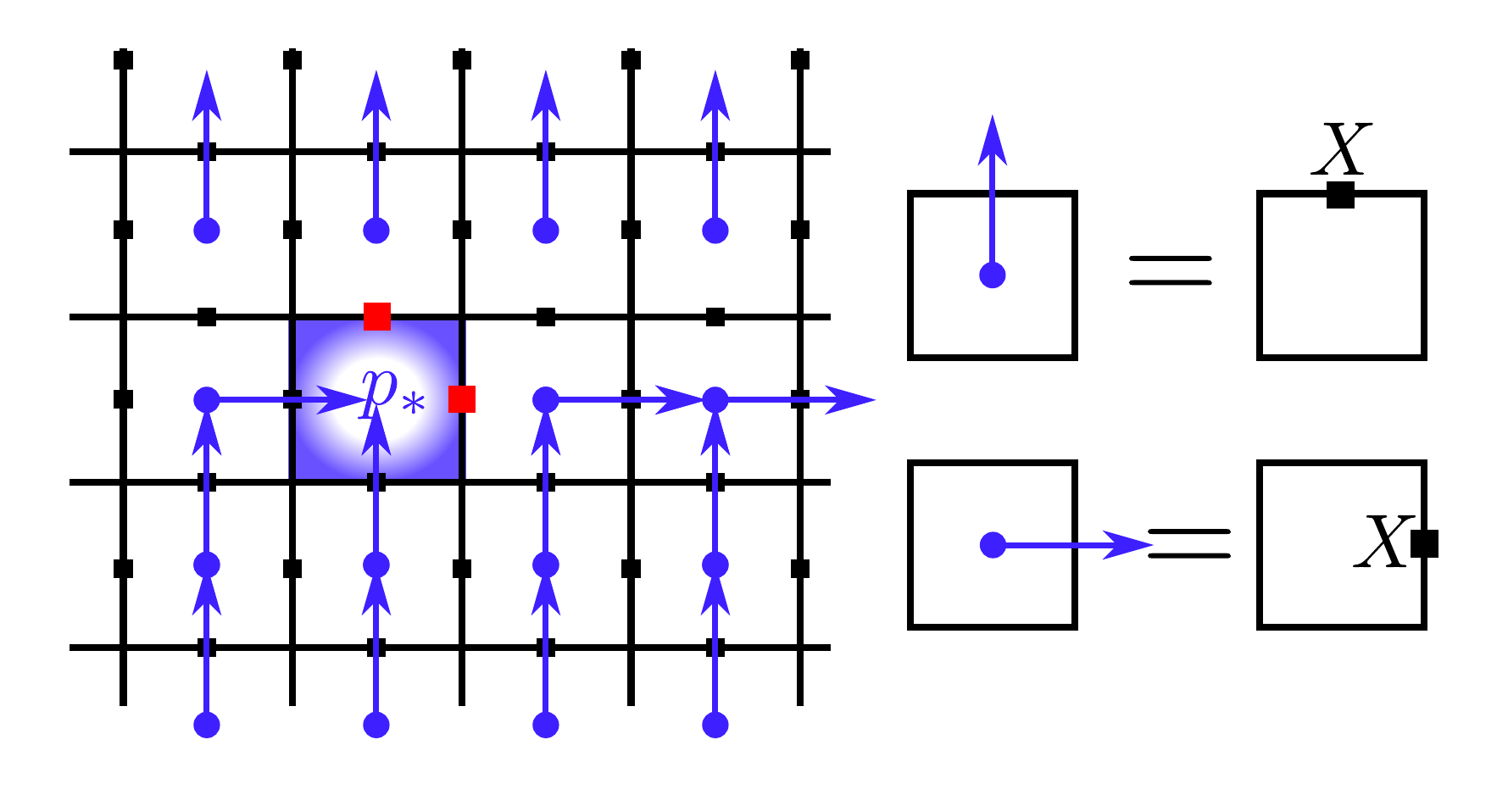}}
\subfigure[Correction operations for vertices\label{it:vertexfig}]{\includegraphics[width=8cm]{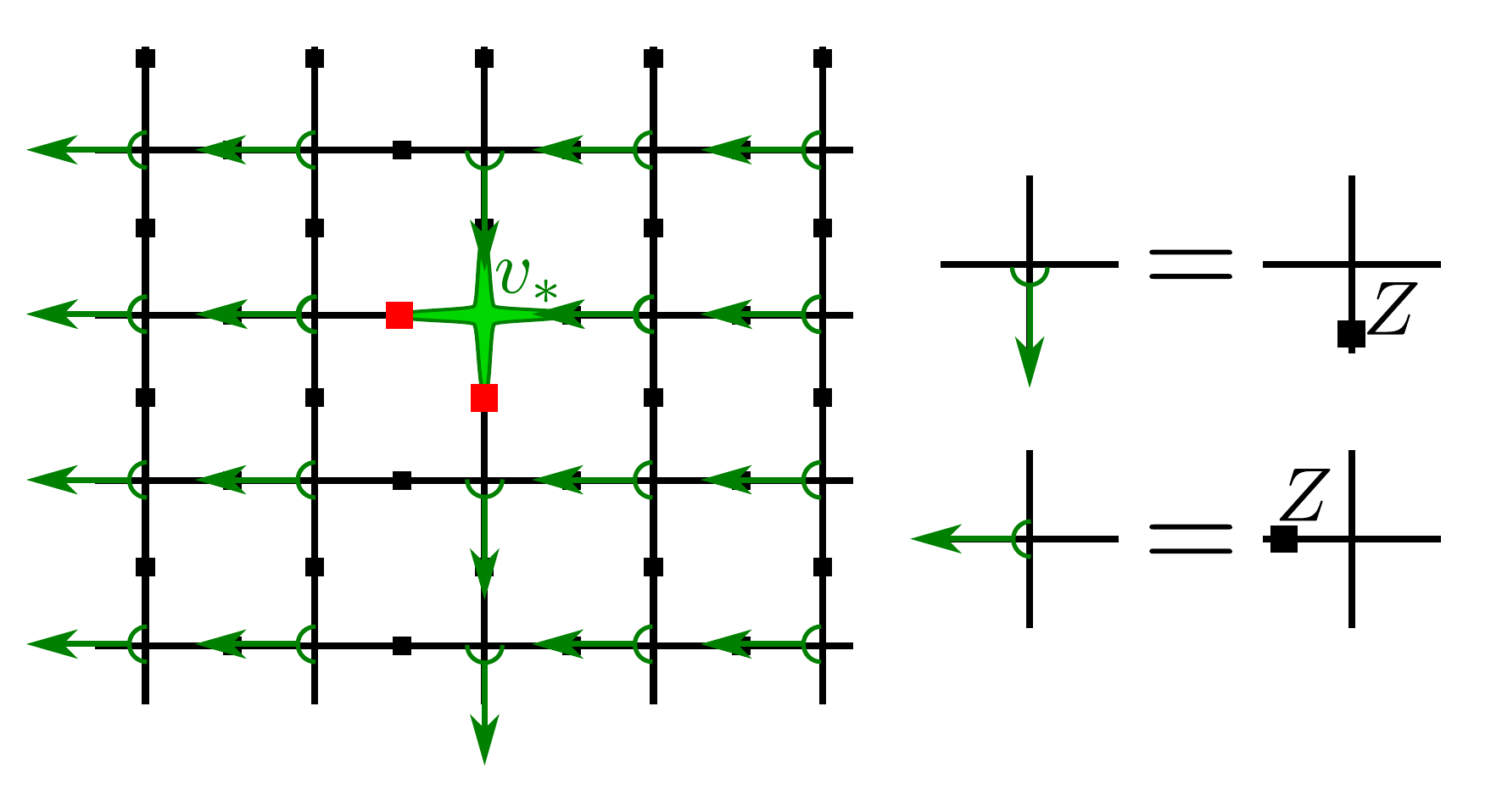}}
\caption{This figure defines the Liouvillian~$\cL$ (see text). In particular,~\ref{it:plaquettesfig} defines the sets $\cS^{(p)}_\rightarrow$ and $\cS^{(p)}_\uparrow$. The associated correction operation~$C_p$ consists of a single-qubit Pauli-$X$. It moves magnetic (plaquette-type) excitations  to the neighboring plaquette according to the indicated arrow. Similarly,~\ref{it:vertexfig} defines the sets $\cS^{(v)}_\leftarrow$ and $\cS^{(v)}_\downarrow$. Electric (vertex-type) excitations are moved from one vertex to the next according to these arrows. The associated correction operator $C_v$ is a single-qubit Pauli-$Z$. The qubits $A_1A_2$ carrying the logical information  are indicated in red. They both are part of the special plaquette~$p_*$, and incident to the vertex~$v_*$. No correction operation acts on the qubits~$A_1A_2$. 
 \label{fig:liouvilleandef}}
\end{figure}

The subscript associated with these sets indicates the direction of movement of an excitation under application of the local correction map. 
For example, a magnetic excitations (caused by an $X$-error) on a plaquette~$p\in\cS^{(p)}_{\uparrow}$ will move to the neighboring plaquette to the north of~$p$ under application of the correction map~$C_p$. That is, we define
\begin{eqnarray}
C_p&=\panyonup \qquad\textrm{ for }\qquad p\in\cS^{(p)}_\uparrow\qquad C_p=\panyonri \qquad\textrm{ for }\qquad p\in\cS^{(p)}_\rightarrow\nonumber\\
C_v&=\vanyondn \qquad\textrm{ for }\qquad v\in\cS^{(v)}_\downarrow\qquad C_v=\vanyonle \qquad\textrm{ for }\qquad v\in\cS^{(v)}_\leftarrow\ .\nonumber
\end{eqnarray}
We will set
\begin{eqnarray*}
C_{v_*}=C_{p_*}=I\ ,
\end{eqnarray*}
corresponding to a trivial correction operation (this is simply done for convenience). The Liouvillian~$\cL$ is then defined as in Section~\ref{sec:generalstabilizercodes}. 

\section{Proof of Theorem~\ref{thm:encoder}}
The proof of Theorem~\ref{thm:encoder} relies on  two basic statements. The first one concerns the logical information encoded in the state: for suitable initial states, this information is preserved along the evolution.
Recall that the toric code has two encoded qubits. Let
\begin{eqnarray}
\bar{X}_1&=\left(\otimes_{b\in\cB} X_b\right)\otimes X_{A_1}\qquad \bar{Z}_1=\left(\otimes_{c\in\cC} Z_c\right)\otimes Z_{A_1}\qquad \bar{Y}_1=i\bar{X}_1\bar{Z}_1\nonumber\\
\bar{X}_2&=\left(\otimes_{b\in\cB'} X_b\right)\otimes X_{A_2}\qquad \bar{Z}_2=\left(\otimes_{c\in\cC'} Z_c\right)\otimes X_{A_2}\qquad \bar{Y}_2=i\bar{X}_2\bar{Z}_2\nonumber
\end{eqnarray}
 be the two-qubit logical Pauli operators defined by Figure~\ref{fig:encoderqubits}, and let us set
\begin{eqnarray*}
\bar{\cP}_2=\{ \bar{P}_1\bar{P}_2\ |\ \bar{P}_1\in \{I,\bar{X}_1,\bar{Y}_1,\bar{Z}_1\}, \bar{P_2}\in \{I, \bar{X}_2,\bar{Y}_2,\bar{Z}_2\}\}
\end{eqnarray*} where $I$ is the identity operator. These operators play a crucial role in the following statement.

\begin{lemma}[Preservation of logical information]\label{lem:logicalinfopres}
For any two-qubit Pauli operator $P=P_1\otimes P_2$, $P_j\in \{I,X_j,Y_j,Z_j\}$, let $\overline{P}=\overline{P}_1\overline{P}_2$, $\overline{P}_j\in\bar{\cP}_2$ be its corresponding logical counterpart. 
Consider an initial state of the form
\begin{eqnarray}\rho_0&=\proj{\Psi}_{\cA}\otimes\proj{+}^{\otimes 2(L-1)}_{\cB\cB'}\otimes\proj{0}^{\otimes 2(L-1)}_{\cC\cC'}\otimes\rho_{\cD}\ ,\label{eq:specialinitialstate}
\end{eqnarray}
where the state~$\rho_\cD$ is arbitrary on $(\mathbb{C}^2)^{\otimes 2(L-1)^2}$. 
Then
\begin{eqnarray}
\tr(\overline{P}e^{t\cL}(\rho_0))&=\bra{\Psi}P\ket{\Psi}\qquad\textrm{ for all }t\textrm{ and all }P\in\cP_2 .\label{eq:toproveeqtimeevolvelog}
\end{eqnarray} 
\end{lemma}
\begin{proof}
Let $P=P_1\otimes P_2$ be arbitrary. Observe first that the rhs of~\eqref{eq:toproveeqtimeevolvelog} is equal to
\begin{eqnarray}
\bra{\Psi}P\ket{\Psi}=\tr(\overline{P}\rho_0)\  \label{eq:logicalopeval}
\end{eqnarray}
because of the definition of~$\overline{P}$ (cf.~Fig.~\ref{fig:liouvilleandef}) and the form of~$\rho_0$, i.e., the fact that the qubits 
in~$\cB\cB'$ and $\cC\cC'$ are $+1$-eigenstates of the single-qubit $X$ and $Z$- operators, respectively. In other words, it suffices to show the expectation value~$\tr(\overline{P}e^{t\cL}(\rho_0))$ is time-independent for initial states~$\rho_0$ of the form~\eqref{eq:specialinitialstate}. We claim that an even stronger statement holds: we have
\begin{eqnarray}
(e^{t\cL})^\dagger (\overline{P})=\overline{P}\qquad\textrm{ for any }\overline{P}\in\overline{\cP}_2\ ,\label{eq:constantnonevolution}
\end{eqnarray}
i.e., any observables of the form~$\overline{P}$ does not evolve in the Heisenberg picture.

To prove~\eqref{eq:constantnonevolution}, observe that the one-parameter family of unital maps $\{(e^{t\cL})^\dagger\}_{t\geq 0}$ is generated by the adjoint~$\cL^\dagger$ of the Liouvillian, hence~\eqref{eq:constantnonevolution} is equivalent to
\begin{eqnarray*}
\cL^\dagger(\overline{P})=0\ .
\end{eqnarray*}
Since $\cL=\sum_{j}\cL_j$ is a sum of Liouvillians~$\cL_j=\cT_j-\id$ associated with stabilizers~$j$,  it suffices to verify that for every $j$, we have $\cT_j^\dagger(\overline{P})=\overline{P}$
or
\begin{eqnarray}
\P_j^+\overline{P}\P_j^+ +\P_j^-C_j^\dagger \overline{P} C_j\P_j^-=\overline{P}\ .\label{eq:pjp}
\end{eqnarray}
Because $\overline{P}$ is a logical operator, it commutes with each stabilizer, and hence also the projections~$\P_j^\pm$. Equality,~\eqref{eq:pjp} is in fact implied by $[C_j^\dagger,\overline{P}]=0$, which we can verify by a case-by-case analysis for single-qubit (logical) operators $\overline{P}\in \{\bar{X}_1,\bar{Z}_1,\bar{X}_2,\bar{Z}_2\}$ (the general case then follows since a product $\overline{P}_1\overline{P}_2$ commutes with $C_j^\dagger$ if each factor $\overline{P}_j$ does). 
 Consider for example the case where $\overline{P}=\bar{Z}_1$. We then have to consider two cases:
\begin{enumerate}[(i)] 
\item the correction~$j=p$ is associated with a plaquette. In this case $C_p$ is a Pauli-$X$ (or the identity). However, inspection of the support of~$\bar{Z}_1$ (see Fig.~\ref{fig:encoderqubits}) 
and the location of the correction operators (see Fig.~\ref{it:plaquettesfig}) reveals that no correction operation $C_p$ acts on the support of $\bar{Z}_1$, hence $[C_p^\dagger,\bar{Z}_1]=0$ as desired.
\item the correction $j=v$ is associated with a vertex. In this case $C_v$ is a Pauli-$Z$ (or the identity), hence $[C_v^\dagger,\bar{Z}_1]=0$ holds trivially.
\end{enumerate}
\end{proof}

The second fundamental statement is about convergence to the ground space.
\begin{lemma}[Convergence time]\label{lem:convergencetimebound}
Let $\gs$ be the projection onto the code space of the toric code, and let $\gsc=I-\gs$ the projection onto the orthogonal complement. 
Let~$\rho$ be an arbitrary initial state. Then we have 
\begin{eqnarray}
\tr(\gsc e^{t\cL}(\rho))\leq \epsilon\qquad\textrm{ for all }t\geq (4\ln(2))\cdot L + 2\ln(1/\epsilon) \ .\label{eq:continuoustimeoverlapbound}
\end{eqnarray}
\end{lemma}
\noindent We have not optimized the constants in this bound as we are interested in the overall (linear) scaling in~$L$. The proof strategy is different from the arguments in~\cite{verstraetewolfcirac09} and may be of independent interest.
\newcommand*{\pred}{\mathsf{Pred}}
\begin{proof}
Consider the function $f:\cS\backslash\{v_*,p_*\}\rightarrow\mathbb{N}\cup \{0\}$ defined as follows: 
\begin{eqnarray*}
\textrm{ for a plaquette $p$: }&f(p) +1 =\textrm{ length of a path  moving north-east from $p$ to }p^* \\ 
\textrm{ for a vertex $v$: } &f(v) +1 =\textrm{ length of a path  moving south-west from $v$ to }v^*\ .
\end{eqnarray*}
In other words, the function 
expresses the axial distance to $v^*$ and $p^*$, respectively. 
For a vertex~$v$ (plaquette~$p$), the quantity~$f(v)$  ($f(p)$) is the number of vertices (plaquettes) traversed by an electric (magnetic) excitation on the primary (dual) lattice before reaching~$v_*$ ($p_*$) along a path $(v=v_1,v_2,\ldots,v_{f(v)},v_{*})$ (or $(p=p_1,p_2,\ldots,p_{f(p)},p_*)$). 
For the special vertex~$v^*$ and the plaquette~$p^*$, we set $f(v^*)=f(p^*)=-1$ for convenience (alternatively, we could omit the discussion of the corresponding trivial correction operations altogether as the stabilizer generators are linearly dependent). 

The key property of the function~$f$ is the fact that it is compatible with the way excitations are propagated under the correction operations. More precisely, for each stabilizer~$S_j$,  let $\pred(j)$ be the set of correction operations that anticommute with it, i.e., 
\begin{eqnarray*}
\pred(j):=\{k\in\cS \ |\ k\neq j \textrm{ and } \{C_k,S_j\}_+ =0\}\ .
\end{eqnarray*}
Then $f$ has the property that 
\begin{eqnarray}
k\in\pred(j)\textrm{ implies } f(k) \geq f(j)+1\ .\label{eq:monotonicitypropertyf}
\end{eqnarray}
Namely, whenever an excitation is created by $C_k$ at $S_j$, one can be certain that a higher valued excitation has been removed at $S_k$.
We will argue that for any $\alpha\geq 1$, we have
\begin{equation}
\tr(\gsc e^{t\cL}(\rho)) \leq  e^{-(1-\alpha^{-1}m)t}\sum_{j\in\cS}\alpha^{f(j)}\ ,\label{eq:upperboundonorthogprojection}
\end{equation}
where $m:=\max_{k\in\cS}|\{ j | k \in \pred(j) \}|$ is the maximal number  of stabilizers a single correction operator can excite.  
Note that Eq. (\ref{eq:upperboundonorthogprojection}) holds for arbitrary stabilizer codes with the corresponding definitions. Specializing to the toric code, where $m=1$, we obtain
$\tr(\gsc e^{t\cL}(\rho))\leq e^{-t/2}\cdot 2^{2L}$ (implying the claim) by choosing  $\alpha=2$ and observing that
\begin{eqnarray*}
\sum_{j\in\cS}\alpha^{f(j)} &=  2\alpha^{-1}\sum_{r,s=0}^{L-1}\alpha^{r+s}=2\alpha^{-1}\left(\frac{\alpha^L-1}{\alpha-1}\right)^2 < \frac{2\alpha^{2L-1}}{(\alpha-1)^2}\Big|_{\alpha=2}  = 2^{2L}\  .
\end{eqnarray*}
 To prove~\eqref{eq:upperboundonorthogprojection}, consider the observable
\begin{eqnarray*}
D&=\sum_{j\in\cS\backslash\{v_*,p_*\}}\alpha^{f(j)} \P_j^-\ .
\end{eqnarray*}
Clearly $D\geq \gsc$ for any $\alpha\geq 1$,  hence it suffices to show that the expectation value $\tr\left( D e^{t\cL}(\rho)\right)$ is upper bounded by the rhs.~of~\eqref{eq:upperboundonorthogprojection}, that is 
\begin{eqnarray}
\tr(De^{t\cL}(\rho))\leq e^{-(1-\alpha^{-1}m)t}\sum_{j\in\cS}\alpha^{f(j)}\ .\label{eq:ddecay}
\end{eqnarray}
Consider the Heisenberg evolution of the projection operators~$\{\P_j\}_j$.
Since Pauli operators either commute or anticommute, a straightforward calculation gives
\begin{eqnarray}
\cL^\dagger_k(\P_j^-)&=\cases{0& if $[C_k,S_j]=0$\\
(I-\P_j^-)\P_k^--\P_j^-\P_k^-& if $\{C_k,S_j\}=0$.\\}\nonumber
\end{eqnarray}
In particular, the expectation values behave classically under the designed Liouvillian, that is (writing $\langle X\rangle_t=\tr(Xe^{t\cL}(\rho))$ for brevity),
\begin{eqnarray}
\frac{d\langle \P_j^- \rangle_t }{dt}&= -\langle \P_j^- \rangle_t  +\sum_{k \in \textsf{Pred}(j)}  \langle(I-\P_j^-) \P_{k}^-\rangle_t-\langle \P_k^- \P_j^-  \rangle_t \nonumber\\
&\leq -\langle \P_j^- \rangle_t + \sum_{k \in \textsf{Pred}(j)}\langle \P_k^- \rangle_t\ . \label{ineq:anyonOccupations}
\end{eqnarray}
According to~\eqref{ineq:anyonOccupations}, we have
\begin{eqnarray*}
\frac{d\langle D \rangle_t}{dt}&\leq - \langle D\rangle_t+\sum_j \alpha^{f(j)}\sum_{k\in\mathsf{Pred}(j)} \langle\P_k^-\rangle_t\\
&\leq - \langle D\rangle_t+\alpha^{-1}\sum_k \alpha^{f(k)} \langle\P_k^-\rangle_t \sum_{j : k\in\mathsf{Pred}(j)} 1
\end{eqnarray*}
where we used property~\eqref{eq:monotonicitypropertyf} on the function~$f$. 
According to the definition of $m$ and $D$, this implies
\begin{eqnarray*}
\frac{d\langle D\rangle_t}{dt}&\leq -(1-\alpha^{-1}m)\langle D\rangle_t\ ,
\end{eqnarray*}
i.e., the expectation value decays exponentially.
The claim~\eqref{eq:ddecay} then follows from the fact that $\P_j^-\leq I$ for all $j\in\cS$ since these are projections, hence $\tr(D\rho)=\langle D\rangle_0\leq \sum_{j\in\cS}\alpha^{f(j)}$.

\end{proof}

With Lemma~\ref{lem:logicalinfopres} and Lemma~\ref{lem:convergencetimebound}, we are ready  to prove our main result.

\begin{proof}[Proof of Theorem~\ref{thm:encoder}]
Consider an initial state $\rho_0$ of the form~\eqref{eq:specialinitialstate} and assume that~$t\geq (4\ln(2))\cdot L + 2\ln(\epsilon^{-1})$.
By Lemma~\ref{lem:convergencetimebound}, we have 
$\tr(\gsc e^{t\cL}(\rho_0))\leq \epsilon$. With the gentle measurement lemma (see e.g.,~\cite[Lemma 9.4.1]{Wilde2013}), this implies 
\begin{eqnarray}
\|e^{t\cL}(\rho_0)-\bar{\rho}'_t\|_1\leq 2\sqrt{\epsilon}\ 
 \ ,\label{eq:tclrhoevolve}
\end{eqnarray}
where $\bar{\rho}'_t$ is defined as
$\bar{\rho}'_t:= \frac{\gs e^{t\cL}(\rho_0)\gs}{\tr(\gs e^{t\cL}(\rho_0))}$.
 Observe that the state~$\bar{\rho}'_t$ is  supported entirely on the code space~$\cQ$.

Note that for any two states~$\rho,\sigma$ we have 
$\|\rho-\sigma\|_1=\max_{\| P \| \leq 1}\tr(P(\rho-\sigma))$  and therefore
$|\tr(P(\rho-\sigma))|\leq\|\rho-\sigma\|_1$ for any normalized operator~$P$. 
Hence~\eqref{eq:tclrhoevolve} implies that for any normalized logical operator $\overline{P}$, we have
\begin{eqnarray}
\big|\tr\left(\overline{P}(e^{t\cL}(\rho_0)-\bar{\rho}_t')\right)\big|\leq 2\sqrt{\epsilon}\ .\label{eq:overlpexp}
\end{eqnarray}
Let $\ket{\overline{\Psi}}\in\cQ$ be the target encoded state, i.e., the state satisfying (cf.~\eqref{eq:logicalopeval} and~\eqref{eq:constantnonevolution})
\begin{eqnarray*}
\bra{\overline{\Psi}}\overline{P}\ket{\overline{\Psi}}=\bra{\Psi}P\ket{\Psi}=\tr(\overline{P}\rho_0)=\tr(\overline{P}e^{t\cL}(\rho_0))\qquad\textrm{ for any logical operator }\overline{P}\ .
\end{eqnarray*}
Combining this with~\eqref{eq:overlpexp}  shows that
$\proj{\overline{\Psi}}$ and $\bar{\rho}'_t$ have approximately the same expectation values for logical operators.
Since both states are supported on the ground space $\cQ$ for which we have access to a full algebra of logical observables (linear combinations of logical Pauli observables), one may choose\footnote{If $P$ is the normalized local observable which optimally distinguishes between the unencoded states $U^\dagger \proj{\overline{\Psi}} U$ and $U^\dagger \bar{\rho}'_t U$, then $\bar{P}=U^\dagger P U$ is the required logical operator.} a normalized logical operator $\overline{P}$ such that $\overline{P}\gs =\sgn(\proj{\overline{\Psi}}-\bar{\rho}'_t)$
which achieves the maximization defining $1$-norm distance
\begin{eqnarray}
\fl \big\|\proj{\overline{\Psi}}-\bar{\rho}'_t\big\|_1 = 
\tr[\sgn(\proj{\overline{\Psi}}-\bar{\rho}'_t)(\proj{\overline{\Psi}}-\bar{\rho}'_t)]=
\big|\tr(\overline{P}(\proj{\overline{\Psi}}-\bar{\rho}'_t))\big| \leq 2\sqrt{\epsilon}\ .\label{eq:closetpsi}
\end{eqnarray}

Combining~\eqref{eq:closetpsi} with~\eqref{eq:tclrhoevolve} and using the triangle inequality, we conclude that 
\begin{eqnarray*}
\|e^{t\cL}(\rho_0)-\proj{\overline{\Psi}}\|_1\leq 4\sqrt{\epsilon}\qquad \textrm{ if }t\geq (4\ln(2))\cdot L+2\ln (\epsilon^{-1})\ .
\end{eqnarray*}
The claim of Theorem~\ref{thm:encoder} follows immediately from this statement.
\end{proof}

\section{Conclusions and outlook}
We have presented a quasi-local time-independent Liouvillian which 
 generates a dissipative encoder: it encodes two physical qubits into the ground space of the toric code. Its key features are  translation-invariance in the bulk, optimal locality of the Lindblad operators, as well as optimal, i.e., linear convergence (encoding) time of the resulting dissipative encoder as a function of the lattice size. This illustrates the power of  engineered dissipation for state preparation in a non-trivial setting with degeneracy. 

Since the Liouvillian is constituted of quasi-local CPTPM associated with each stabilizer, one may interpret\footnote{To do so, simply expand the exponential.}  the resulting evolution $e^{t\cL}$ as the repeated application of the `average' correction map (cf.~\eqref{eq:averagebtsuperop}). This viewpoint immediately results in a preparation algorithm for surface code states. The latter inherits the properties of the dissipative evolution: it is translation-invariant up to boundary conditions,  composed of feedback-free quasi-local maps, and requires only a minimal number of iterations to converge.  These features make it an attractive candidate for potential experimental realizations.

\subsection*{Towards fault tolerance}

The most pressing open question is to incorporate noise in the form of control and initialization imperfections into the design and analysis of such an encoder. 
The holy grail is to engineer a reasonable, possibly time-dependent Liouvillian Master equation such that after a specified initial period physical qubits are encoded with only a constant, lattice-size independent probability of suffering uncorrectable errors. 
Expecting further reduction seems impossible without changing the assumptions due to the possibility of an error occuring on the unencoded qubit.
As it stands, our protocol $L$-fold amplifies such an initial error probability as certain errors on the qubits in $\mathcal{B} \cup \mathcal{C} \cup \mathcal{B}' \cup \mathcal{C}'$ are propagated into logical errors.
Indeed, we suspect that a fully fault-tolerant encoding scheme would require progressively  growing the code's physical support while simultaneously introducing dissipative terms responsible for self-correction within this support.
As in \cite{Pastawskietalmemory11} such terms would be responsible for ensuring that the rate of occurence for uncorrectable errors decreases with the code size.

\subsection*{Acknowledgments}
JD and RK gratefully acknowledge funding provided by NSERC. 
RK would like to thank John Preskill for discussions, and the Institute for Quantum Information and Matter for their hospitality. 
FP would like to thank Iman Marvian, Ignacio Cirac and John Preskill for helpful discussions and comments.
FP acknowledges funding provided by the Institute for Quantum Information and Matter, a NSF Physics Frontiers Center with support of the Gordon and Betty Moore Foundation (Grants No. PHY-0803371 and PHY-1125565).

\section*{References}

\begin{thebibliography}{10}

\bibitem{AguadoVidal08}
M.~Aguado and G.~Vidal.
\newblock Entanglement renormalization and topological order.
\newblock {\em Phys. Rev. Lett.}, 100(070404), 2008.

\bibitem{Bardynetal13}
C.-E. Bardyn, M.~A. Baranov, C.~V. Kraus, E.~Rico, A.~Imamoglu, P.~Zoller, and
  S.~Diehl.
\newblock Topology by dissipation, 2013.
\newblock arXiv:1302.5135.

\bibitem{Bravyietalpropagationtop}
S.~Bravyi, M.~B. Hastings, and F.~Verstraete.
\newblock Lieb-{R}obinson bounds and the generation of correlations and
  topological quantum order.
\newblock {\em Phys. Rev. Lett.}, 97(050401), 2006.

\bibitem{Brown2011}
Benjamin~J Brown, Wonmin Son, Christina~V Kraus, Rosario Fazio, and Vlatko
  Vedral.
\newblock {Generating topological order from a two-dimensional cluster state
  using a duality mapping}.
\newblock {\em New Journal of Physics}, 13(6):065010, June 2011.

\bibitem{dennisetal02}
E.~Dennis, A.~Kitaev, A.~Landahl, and J.~Preskill.
\newblock Topological quantum memory.
\newblock {\em J. Math. Phys.}, 43:4452--4505, 2002.

\bibitem{DiehlMicheli08}
S.~Diehl, A.~Micheli, A.~Kantian, B.~Kraus, H.~P. B{\"u}chler, and P.~Zoller.
\newblock Quantum states and phases in driven open quantum systems with cold
  atoms.
\newblock {\em Nature Physics}, 4:878--883, September 2008.

\bibitem{Grudkaetal12}
A.~Grudka, M.~Horodecki, P.~Horodecki, P.~Mazurek, L.~Pankowski, and
  A.~Przysiezna.
\newblock Long distance quantum communication over noisy networks without
  quantum memory, 2012.
\newblock arXiv:1202.1016.

\bibitem{Katoryanotiming}
M.~J. Kastoryano, M.~M. Wolf, and J.~Eisert.
\newblock Precisely timing dissipative quantum information processing.
\newblock {\em Phys. Rev. Lett.}, 110:110501, Mar 2013.

\bibitem{kitaev97}
A.~Y. Kitaev.
\newblock Fault-tolerant quantum computation by anyons.
\newblock {\em Ann. Phys.}, 303:2--30, 2003.

\bibitem{KoePas13}
R.~K\"onig and F.~Pastawski.
\newblock Generating topological order: no speedup by dissipation, 2013.
\newblock arXiv:1310.1037.

\bibitem{KrausDiehlZoller08}
B.~Kraus, H.~P. B\"uchler, S.~Diehl, A.~Kantian, A.~Micheli, and P.~Zoller.
\newblock Preparation of entangled states by quantum {M}arkov processes.
\newblock {\em Phys. Rev. A}, 78:042307, Oct 2008.

\bibitem{Marvian2013}
I.~Marvian.
\newblock Personal communication, 2013.

\bibitem{Pastawskietalmemory11}
F.~Pastawski, L.~Clemente, and J.~I. Cirac.
\newblock Quantum memories based on engineered dissipation.
\newblock {\em Phys. Rev. A}, 83:012304, Jan 2011.

\bibitem{PEPSinjective}
D.~Perez-Garcia, F.~Verstraete, J.~I. Cirac, and M.~M. Wolf.
\newblock P{E}{P}{S} as unique ground states of local {H}amiltonians.
\newblock {\em Quantum Information \& Computation}, 8(6):0650--0663, July 2008.

\bibitem{TicozziViola12}
F.~Ticozzi and L.~Viola.
\newblock Stabilizing entangled states with quasi-local quantum dynamical
  semigroups.
\newblock {\em Phil. Trans. R. Soc. A}, 370(1979):5259--5269, October 2012.

\bibitem{Ticozzi2014}
F.~Ticozzi and L.~Viola.
\newblock {Steady-State Entanglement by Engineered Quasi-Local Markovian
  Dissipation}.
\newblock {\em Quantum Information and Computation}, 14(3-4):0265--0294, April
  2014.

\bibitem{Tomadinetal12}
A.~Tomadin, S.~Diehl, M.~D. Lukin, P.~Rabl, and P.~Zoller.
\newblock Reservoir engineering and dynamical phase transitions in
  optomechanical arrays.
\newblock {\em Phys. Rev. A}, 86:033821, Sep 2012.

\bibitem{verstraetewolfcirac09}
F.~Verstraete, M.~M. Wolf, and J.~I. Cirac.
\newblock Quantum computation and quantum-state engineering driven by
  dissipation.
\newblock {\em Nature Physics}, 5:633--636, July 2009.

\bibitem{Wilde2013}
M.~Wilde.
\newblock {\em {Quantum Information Theory}}.
\newblock Cambridge University Press, 2013.

\bibitem{WolfCirac06}
M.~Wolf and J.~I. Cirac.
\newblock Dividing quantum channels.
\newblock {\em Comm. Math. Phys.}, 279(147), 2008.

\end{thebibliography}

\end{document}